\documentclass[letterpaper]{article}
\usepackage{uai2019}
\usepackage[margin=1in]{geometry}

\usepackage{times}
\usepackage{natbib}
\usepackage{breakcites}
\usepackage{graphicx}
\usepackage{amsmath}
\usepackage{amsfonts}
\usepackage{amssymb}
\usepackage{amsthm}
\usepackage{xcolor}
\usepackage{caption} 

\usepackage{subcaption}
\usepackage{booktabs}
\usepackage[title]{appendix}
\usepackage{hyperref}
\usepackage{breakurl}

\def\clique{\mathcal{C}}
\def\Prob{\mathbf{P}}
\def\nc{N}
\def\cc{n}

\def\levy{\nu}
\def\atom{\mu}

\def\Pois{\mbox{Poisson}}

\def\stablebeta{\mbox{SB}}

\def\E{\mathbb{E}}
\def\numtriangles{\lambda}
\def\clusteringcoeff{C}

\newtheorem{proposition}{Proposition}
\newtheorem{corollary}{Corollary}
\title{Random clique covers for graphs with local density and global sparsity}

\author{Sinead A. Williamson$^{1, 2, 3}$ and Mauricio Tec$^1$\\ $^1$Department of Statistics and Data Science, University of Texas at Austin \\ $^2$Department of Information, Risk and Operations Management, University of Texas at Austin\\ $^3$CognitiveScale}
\begin{document}
\maketitle

\begin{abstract}
 Large real-world graphs tend to be sparse, but they often contain many densely connected subgraphs and exhibit high clustering coefficients. While recent random graph models can capture this sparsity, they ignore the local density, or vice versa. We develop a Bayesian nonparametric graph model based on random edge clique covers, and show that this model can capture power law degree distribution, global sparsity and non-vanishing local clustering coefficient. This distribution can be used directly as a prior on observed graphs, or as part of a hierarchical Bayesian model for inferring latent graph structures.

\end{abstract}

\section{INTRODUCTION}
Random graph models provide statistical tools for network analysis and can often be used as prior distributions in a Bayesian framework. Such models aim to capture various properties of real-world graphs, such as power-law degree distributions \citep{Albert:Barabasi:2002,Dan2004,bloem2016}, small-world properties \citep{WatStr1998}, or latent community structure \citep{HolLasBlaLei1983,AirBleFieXin2008,KarNew2011}.

One statistic of interest is the density of a graph, defined as the number of edges over the number of \textit{possible} edges for a binary, undirected graph. This can be extended to distributions over graphs: we think of a random graph as being dense if the number of edges grows quadratically with the number of vertices, and sparse if it grows sub-quadratically with the number of vertices \citep{NesOss2012}. Many commonly-used random graphs, such as Erd\"{o}s-R\'{e}nyi graphs \citep{Gil1959,ErdRen1960} or stochastic blockmodels and their variants \citep{HolLasBlaLei1983,AirBleFieXin2008,KarNew2011}, concentrate on dense graphs \citep[see][for a discussion]{LloOrbGhaRoy2012}. Dense behavior is commonly seen in small, closed communities, where every vertex has the option to interact with all other vertices and a fully connected graph is at least conceivable. However it is not generally seen in larger graphs, where a given vertex will not have the opportunity to interact with more than a small subset of the other vertices---for example international online social networks. 

 Recently, a number of Bayesian nonparametric models have been proposed that can generate sparse graphs. The sparse exchangeable graph framework, developed by \citet{CarFox2017} and further explored by  \citet{VeiRoy2015} and \citet{BorChaHol2018}, constructs sparse graphs or multigraphs based on a random atomic measure on the space $\mathcal{V}$ of potential vertices. 
 The edge exchangeable graph framework \citep{CaiCamBro2016,CraDem2017} constructs sparse multigraphs as a sequence of edges, whose distribution is parametrized by a random atomic measure on $\mathcal{V}$. 
 In both classes of model, appropriate design choices yield power-law degree  distribution, another statistical property often seen in real-world graphs.

While distributions such as these offer the ability to model sparse graphs, they do not capture certain other important graph behaviors. While large-scale graphs are usually sparse, locally they exhibit transitivity: if two vertices $j$ and $k$ are both connected to vertex $i$, a connection between vertices $j$ and $k$ is more likely than a connection between $j$ and a randomly selected vertex. This means that the graph will tend to exhibit densely connected subgraphs, even if the overall graph is sparse. Such graphs will typically exhibit a high average local clustering coefficient, indicating that most vertices' neighborhoods are close to being cliques and that the graph has a relatively high number of triangles. This type of behavior is a key component of ``small world'' graphs \citep{WatStr1998}. While extensions to the sparse exchangeable graphs \citep{HerSchMor2016,TodMisCar2016} and the edge-exchangeable graphs \citep{Wil2016} incorporate community-type structure, they do not yield this form of transitivity.

In this paper we explore distributions over graphs that, while sparse, exhibit local density in the form of arbitrarily-large fully-connected subgraphs and non-vanishing average local clustering coefficient.  We achieve this by explicitly modeling graphs as collections of cliques, or fully connected subsets of vertices. Our construction takes the form of a random clique cover, with cliques selected using a nonparametric feature selection model (in this paper, the stable beta Indian buffet process \citep{TehGor2009}, but other choices are possible). The resulting graphs exhibit local density and non-zero local clustering coefficients, with clique sizes controlled primarily by a single parameter. On a global scale, the graphs are sparse, with the degree of sparsity primarily controlled by a separate parameter. We observe power-law distributions over both the number of cliques a vertex belongs to, and its degree. 

When used directly as a distribution over graphs, we show that this model's statistical properties mimic the statistics of real-world graphs (Section~\ref{sec:exp_full}). When used in a hierarchical modeling context, the cliques can be thought of as latent communities (Section~\ref{sec:exp_partial}).

\section{RANDOM CLIQUE COVERS}\label{sec:model}

A graph $G$ is most commonly described using an ordered pair $(V,E)$ of vertices $V$ and edges $E$. $G$ can alternatively be described in terms of its edge 
clique cover \citep{Orlin1977}. An edge clique cover (or intersection graph) is a set of cliques---i.e.\ fully connected subgraphs---such that two vertices share an edge iff they have at least one clique in common. 

This formulation suggests a method for generating random graphs by placing a distribution over cliques. Concretely, let a random clique be a random, finite subset $\clique\subset \mathcal{V}$ of a (possibly uncountable) set $\mathcal{V}$ of potential vertices. Given a distribution over sequences of such subsets, we can generate a sequence $\clique_1, \dots, \clique_\nc$ of cliques, that can be translated into a graph $G=(V,E)$  by adding undirected edges between all vertices with at least one clique in common, so that
$$\begin{aligned}
\textstyle (i, j) &\in \textstyle E\; \mbox{ iff } \; \sum_{\cc=1}^\nc I(i,j\in \clique_\cc \; i\neq j) > 0\\
\textstyle i &\in \textstyle V\; \mbox{ iff } \; \sum_{\cc=1}^\nc I(i\in \clique_\cc) > 0\, .\end{aligned}$$

If we number the vertices in our graph, and represent our cliques in terms of a binary matrix $Z$ where $Z_{\cc i}=1$ iff clique $\cc$ includes vertex $i$, then we can represent the graph's adjacency matrix as 
$$\min\left(Z^TZ-\mbox{diag}(Z^TZ), 1\right)$$
where the min is taken as an element-wise operation.

We can modify this construction to give a distribution over random multigraphs, by letting the number of edges between verticies $i$ and $j$ equal the number of cliques they have in common. We discuss this multigraph formulation in Appendix~\ref{app:multigraph}.

 \subsection{RANDOM CLIQUES GENERATED USING EXCHANGEABLE RANDOM MEASURES}\label{sec:stablebeta}
 
 We use the idea of random clique covers to construct a distribution over graphs of fixed (but random) dimensionality. To do so, we much specify both a distribution over the number of cliques, and distributions over the vertices appearing in those cliques. We let the number of cliques $\nc \sim \mbox{Poisson}(\tau)$, although other choices could easily be made. 
 
 Reasonable desiderata for the random clique selection mechanism might be that the total number of vertices (and edges) is unbounded; that the size of each clique is random; and that cliques overlap with finite probability (avoiding the trivial solution where our graph is made up of a series of disconnected subgraphs). 
 
 One choice that meets these criteria is the stable beta process \citep{TehGor2009}. The stable beta process is a completely random measure \citep{Kin1967} whose atoms lie in (0,1], with L\'{e}vy measure
 
 \begin{equation}
     \levy(d\atom) = \alpha\frac{\Gamma(1+c)}{\Gamma(1-\sigma)\Gamma(c+\sigma)}\atom^{-\sigma-1}(1-\atom)^{c+\sigma-1}d\atom \label{eqn:levymeasure}
 \end{equation}
 
 where $\sigma\in (0,1]$, $c\geq -\sigma$, and  $\alpha>0$. If $\sigma=0$, this reduces to the (homogeneous) beta process \citep{Hjo1990,ThiJor2007}. If $c=1-\sigma$, the stable beta process corresponds to a stable process with all atoms of size larger than one removed.

We can use the stable beta process to construct a distribution over cliques, which we can represent in terms of a binary matrix. If $\atom = \sum_{i=1}^\infty \atom_i\delta_{\theta_i} \sim \stablebeta(c, \sigma,\alpha)$ where each location $\theta_i$ corresponds to a vertex, we can sample an exchangeable sequence of $\nc$ cliques by including the $i$th vertex in the $\cc$th clique with probability $\atom_i$. We can represent the resulting clique allocation as a binary matrix $Z$ with $\nc$ rows and infinitely many columns, where $Z_{\cc i} = 1$ iff vertex $i$ is in clique $\cc$.

If we marginalize out $\atom$, we obtain a predictive distribution for $Z_\cc|Z_1,\dots, Z_{\cc-1}$ that is known as the stable-beta Indian buffet process \citep[SB-IBP,][]{TehGor2009}. If $m_i = \sum_{j=1}^{\cc-1}Z_{ji}$ is the number of times vertex $i$ has been previously selected, then the $\cc$th clique will include that vertex (i.e.\ $Z_{\cc i}=1$) with probability
$$\Prob(Z_{\cc i}=1|Z_1,\dots, Z_{\cc-1}) = \frac{m_i-\sigma}{\cc+c-1}\, .$$
In addition to vertices that have previously appeared, $Z_\cc$ will also select $\Pois\left(\alpha \frac{\Gamma(1+c)\Gamma(\cc+c+\sigma-1)}{\Gamma(\cc+c)\Gamma(c+\sigma)}\right)$ new vertices.

The SB-IBP exhibits a number of power-law behaviors \citep{TehGor2009,BroJorPit2012,HeaRoy2015}. 
First, the total number of non-zero columns of $Z$ (in our case, observed vertices) in $\nc$ rows (cliques) follows a power law, growing, as $\nc\rightarrow \infty$, as $\frac{\alpha}{\sigma}\frac{\Gamma(c+1)}{\Gamma(c+\sigma)} \nc^\sigma$.

Second, the number of non-zero entries per column (in our case, the number of cliques each vertex appears in) follows a Zipf's law. Let $K_{\nc,j}$ be the number of features that appear exactly $j$ times in $\nc$ rows. Then, following \citet{BroJorPit2012},
\begin{equation}K_{\nc,j} \sim \frac{\alpha\Gamma(j-\sigma)\Gamma(1+c)}{j!\Gamma(1-\sigma)\Gamma(c+\sigma)}\nc^\sigma
\sim j^{-\sigma-1}\, .\label{eqn:zipf}\end{equation}
As we will see in Section~\ref{sec:properties}, when applied to edge clique covers, the SB-IBP yields many interesting graph-specific properties, many of which follow from these two observations.

\subsection{A MODEL FOR PARTIALLY OBSERVED CLIQUES}\label{sec:ext}
In a modeling context, we may wish to interpret cliques as latent communities \citep{HolLasBlaLei1983,AirBleFieXin2008,KarNew2011}. For example in a social network, a clique might represent a shared interest or hobby. In this context, we would not necessarily expect all vertices in the community to be connected. Instead, we can think of the community as being a noisy instantiation of a fully connected latent clique. 

To model this in the current context, start with the distribution over edge clique covers described in Section~\ref{sec:stablebeta}. If we associate each latent clique $\clique_\cc$ with a probability $\pi_\cc$, we can form a graph by including an edge between vertices $i$ and $j$ with probability

$$P\left((i,j) \in E\right) = 1-\prod_{\cc: i,j \in \clique_\cc}(1-\pi_\cc)\, .
$$

If we use a single probability $\pi$ across all cliques, this corresponds to including edges according to a noisy-OR likelihood: if two vertices have $m$ latent cliques in common, they share an edge with probability $1-(1-\pi)^m$.

One way of thinking of this model is as a superposition of locally-defined Erd\"{o}s-Renyi $G(n,p)$ graphs. A $G(n,p)$ graph builds a graph with $n$ vertices, by including each of the $\binom{n}{2}$ potential edges with probability $p$. In the partially observed clique model, each row of the SB-IBP selects a clique $\clique_\cc$ of vertices, and builds a subgraph on those vertices according to a $G(|\clique_\cc|,\pi_\cc)$ model. Within the clique, we have a dense subgraph almost surely.

\section{GRAPH PROPERTIES}\label{sec:properties}

In this section, we explore the statistical properties of graphs constructed using an SB-IBP distribution over cliques.  In particular, we show that we obtain sparse graphs that exhibit densely connected subgraphs, with a non-vanishing local clustering coefficient. We focus on the fully observed setting described in Section~\ref{sec:stablebeta}, noting that the partially observed extension described in Section~\ref{sec:ext} will inherit similar properties.

\subsection{SPARSITY}\label{sec:sparsity}

Let $|E|$  be the number of edges in a graph, and $|V|$ be the number of vertices. We can define the density of a fixed graph as the number of edges over the number of possible edges, i.e.\ $D = 2|E| /|V|\left(|V|-1\right)$. 
We can extend this concept to random graphs, by looking at how the number of edges varies with the number of vertices as we change parameters of the model. We say a random graph is dense if $|E|$ grows quadratically with $|V|$, and sparse if it grows sub-quadratically. To explore the sparsity of the proposed random graph, we will look at how $|E|$ and  $|V|$ vary with the number of cliques. In our construction, the number of cliques is a $\mbox{Poisson}(\tau)$ random variable, and so the expected number of cliques can be manipulated by altering the rate $\tau$ of that Poisson distribution.

\begin{proposition}\label{prop:vertices}
As the number $\nc$ of cliques grows to infinity, the number of vertices $|V|$ grows as $\nc^\sigma$.
\end{proposition}

\begin{proof}
This is a direct consequence of the power-law behavior of the SB-IBP: the number of vertices in a graph with $\nc$ generating cliques correspond to the number of non-zero columns in a sample from the SB-IBP with $\nc$ rows. \citet{BroJorPit2012} and \citet{HeaRoy2015} show that the expected number of non-zero columns  grows as $\frac{\alpha}{\sigma}\frac{\Gamma(c+1)}{\Gamma(c+\sigma)}\nc^\sigma$.
\end{proof}

\begin{proposition}\label{prop:edges}
As the number $\nc$ of cliques grows to infinity, the number of edges $|E|$ grows as $O\left( \min(N^{\frac{1+\sigma}{2}}, N^{\frac{3\sigma}{2}})\right)$. Combined with the result in Proposition~\ref{prop:vertices}, this means that the number of edges grows with the number of vertices as $O(\left(\min(|V|^{\frac{1+\sigma}{2\sigma}}, |V|^{3/2})\right)$, implying that the graph is sparse.
\end{proposition}

\begin{proof}
Conditioned on the probabilities $\atom_i,\atom_j$ assigned by the stable beta process to vertices $i$ and $j$, the probability of an edge between $i$ and $j$ given $\nc$ cliques is $1-(1-\atom_i\atom_j)^\nc$. By linearity of expectation, the total expected number of edges, given $\atom$, is therefore $\sum_{i>j}\left(1 - (1-\atom_i\atom_j)^\nc\right) $. Therefore, we can obtain the expected total number of edges in the graph as
\begin{equation}
    \E\left[|E|\right] = \frac{1}{2}\int\int \left(1 - (1-wv)^\nc\right) \nu(dw) \nu(dv)\label{eqn:number_of_edges}
\end{equation}
This equation also describes the number of edges in an edge-exchangeable graph based on the stable beta process, proposed by \cite{CaiCamBro2016}. Although the two models are different---our model is not edge-exchangeable, and the \citeauthor{CaiCamBro2016} model is not locally dense---the expected number of edges coincide due to linearity of expection. \citeauthor{CaiCamBro2016} show, via a  Poissonization approach, that the integral in Equation~\ref{eqn:number_of_edges} grows as $O\left(\min(\nc^{\frac{1+\sigma}{2}},\nc^{\frac{3\sigma}{2}})\right)$.\footnote{For results for alternative choices of completely random measure, see also \cite{CaiCamBro2016}.} 
\end{proof}

\begin{corollary}The number of edges grows with the number of vertices as $O(\left(\min(|V|^{\frac{1+\sigma}{2\sigma}}, |V|^{3/2})\right)$, implying that the graph is sparse for all values of $\sigma$.
\end{corollary}

\begin{proof}
This follows directly from Propositions~\ref{prop:vertices} and \ref{prop:edges}, and is also noted by \citet{CaiCamBro2016}.
\end{proof}
 As we approach the limit where $\sigma=1$, we obtain the trivially sparse graph where each clique contains a $\Pois(\alpha)$ number of vertices, and there are no edges connecting cliques. In the supplement, we show  simulations that validate these results empirically.

\subsection{LOCAL DENSITY AND CLUSTERING COEFFICIENTS}\label{sec:density}

By construction, each observed vertex appears in at least one of the $\nc$ generating cliques in $Z$, with expected clique size $\alpha$. As the number of generating cliques a vertex appears in grows, so does the expected size of the largest generating cliques (since $\E[\max(X_1,X_2)]\geq \max(\E[X_1], \E[X_2])$). The expected size of the maximal clique for the vertex (i.e.\ the largest clique containing that vertex) will be larger again, since three or more overlapping cliques can introduce triangles that do not appear in any of the individual cliques. This ensures that, even though the overall graph is sparse, there will be locally dense subgraphs, with expected size lower bounded by $\alpha$. We show plots demonstrating this empirically in the supplement.

The local clustering coefficient \citep{WatStr1998} is a popular tool for measuring the local density of a graph. The local clustering coefficient of vertex $i$ is given by $\clusteringcoeff_i = \numtriangles_i/k_i(k_i-1)$, 
where $\numtriangles_i$ is the number of triangles containing vertex $i$, and $k_i$ is the number of neighbors of vertex $i$. In other words, $\clusteringcoeff_i$ is a local measure of the ratio of the number of triangles to the number of potential triangles, given a vertex's neighborhood. 

The local clustering coefficient is a measure of the transitivity of the graph: to what extent the presence of the edges $(i, j)$ and $(j, k)$ imply the presence of an edge $(i, k)$. Intuitively, adding edges in cliques, rather than individually, increases transitivity: knowing that the graph contains edges $(i, j)$ and $(j, k)$ increases the likelihood that there is at least one clique containing all three vertices, implying the presence of the third edge.

\begin{proposition} The average local clustering coefficient remains non-zero as the number of cliques grows to infinity, provided $\sigma>0$.
\end{proposition}
\begin{proof}
Let $j$ be the number of cliques to which a given vertex belongs. Vertices belonging to a single clique have a local clustering coefficient of 1. As $j$ increases, the expected local clustering coefficient will decrease, and will be minimized when it is the only common member of those $j$ cliques. In this extreme setting, the clustering coefficient will vary as $j^{-1}$ (since the vertex's expected neighborhood size grows as $j\alpha$, and the expected number of triangles as $j\alpha^2$). For large $j$ and $\nc$, the distribution over the number of cliques, $j$, that a vertex belongs to can be approximated by a Zeta($1+\sigma)$ distribution (Equation~\ref{eqn:zipf}), and the expectation of $j^{-1}$ under this distribution is $\zeta(2+\sigma)/\zeta(1+\sigma)>0$ for $\sigma \in (0,1)$. This implies that the average clustering coefficient will not tend to zero as the graph grows.
\end{proof}

\subsection{DEGREE DISTRIBUTION}\label{sec:degreedist}

The expected size of each generating clique is $\alpha$, and the number of cliques an vertex belongs to follows a power law distribution (Equation~\ref{eqn:zipf}). This leads to a power law distribution over the degree, with the expected degree decreasing with $\sigma$ and increasing with $\alpha$.

The left hand column of Figure~\ref{fig:degree_dist_1} shows the degree distribution (in the form of a boxplot) for different values of $\sigma$ and $\alpha$, and different numbers of cliques $\nc$. We see increasingly heavy tails as $\sigma$ increases. 

Since the number of vertices also increases with $\sigma$, the maximum degree is higher for larger $\sigma$. The right hand column of Figure~\ref{fig:degree_dist_1} shows the degree divided by the total number of vertices, making it easier to compare different values of $\sigma$. We see that the average proportion decreases and the distribution grows heavier tailed as $\sigma$ increases, for all values of $\alpha$ and $\nc$.
\begin{figure}[t!]
  \centering
  \begin{subfigure}{0.45\columnwidth}
    \includegraphics[width=\textwidth]{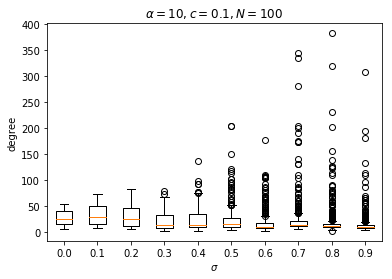}
  \end{subfigure}
    \begin{subfigure}{0.45\columnwidth}
    \includegraphics[width=\textwidth]{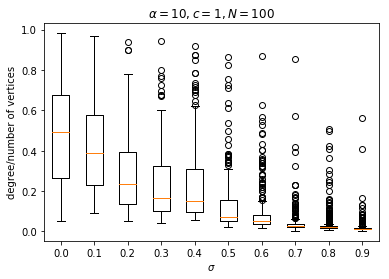}
  \end{subfigure}\\
    \begin{subfigure}{0.45\columnwidth}
    \includegraphics[width=\textwidth]{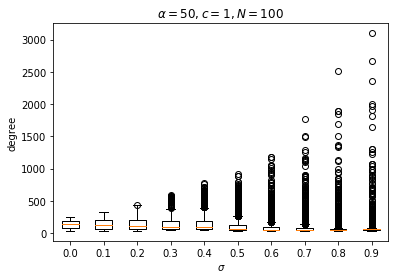}
  \end{subfigure}
    \begin{subfigure}{0.45\columnwidth}
    \includegraphics[width=\textwidth]{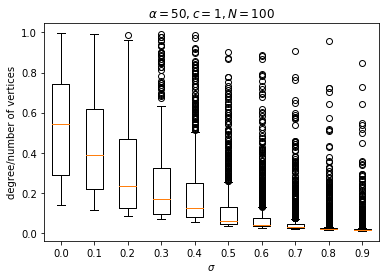}
  \end{subfigure}\\
  \begin{subfigure}{0.45\columnwidth}
    \includegraphics[width=\textwidth]{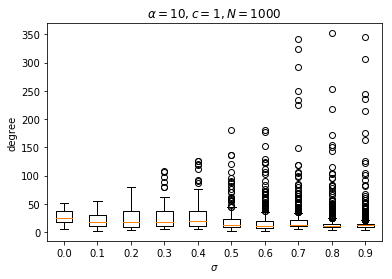}
  \end{subfigure}
    \begin{subfigure}{0.45\columnwidth}
    \includegraphics[width=\textwidth]{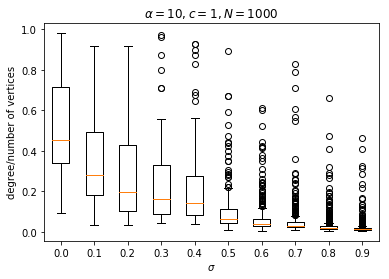}
  \end{subfigure}\\
  \caption{Vertex degrees for different values of $\sigma$, $\alpha$ and $\nc$, for $c=1$. Left: Raw degrees; Right: degree/number of vertices.}
  \label{fig:degree_dist_1}
\end{figure}

\subsection{DENSITY OF THE INTERSECTION GRAPH OF GENERATING CLIQUES}\label{sec:cliquegraph}

The clique graph of a graph is the intersection graph of its maximal cliques \citep{Hamelink:1968}. Since our construction does not explicitly generate maximal cliques, we consider the generating-clique graph---i.e.\ the intersection graph of the generating cliques specified in $Z$. Perhaps surprisingly given the sparsity of the overall graph, this intersection graph is dense.
\begin{proposition}
The intersection graph of the maximal cliques is dense, and the expected number of edges connecting two generating cliques is $\alpha \frac{1-\sigma}{1+c}$.
\end{proposition}
\begin{proof}
The expected number of vertices by which two cliques overlap is given by $\sum_{i=1}^\infty \atom_i^2$. Campbell's theorem tells us that $E\left[\sum_{x\in \Pi}f(x)\right] = \int_S f(x)\nu(dx)$, where $\Pi$ is a Poisson process on $S$ with rate measure $\nu(dx)$---in our case, the L\'{e}vy measure $\nu(d\atom)$ of the stable beta process (Equation~\ref{eqn:levymeasure})---and $f$ is a measurable function. Therefore the expected overlap between two cliques is
$$\int_0^1\nu(d\atom)= \alpha\frac{1-\sigma}{1+c}.$$

This indicates that two cliques overlap with positive probability, meaning that the resulting intersection graph is dense and that there is a path between any two vertices with positive probability, even if they do not belong to the same clique.
\end{proof}

\section{RELATED WORK}\label{sec:related}

The idea of a random edge clique cover was explored by \citet{Bar2008}, under the name Clique Matrices. As in this paper, a binary matrix $Z$ is used to represent cliques within the graph; however the number of vertices is fixed and the entries are i.i.d.\ Bernoulli random variables. Such a graph is a special case of the model proposed here, but it lacks the sparsity, power-law degree, and unbounded number of vertices that are discussed in Section~\ref{sec:properties}. \citet{Bar2008} also proposed a noisily-observed Clique Matrix model, where two vertices are connected with probability $\sigma(Z^TZ)$ (where $\sigma(\cdot)$ represents the sigmoid function). Unlike the noisy-OR formulation for partially observed cliques proposed in Section~\ref{sec:ext}, this model precludes sparse graphs, since it allows edges between vertices with no cliques in common. 

Kronecker graphs \citep{Leskovec:Chakrabarti:Kleinberg:Faloutsos:Ghahramani:2010} are a family of random graph models that construct adjacency matrices by starting with a low-dimensional initiator matrix, and iteratively applying the Kronecker product. The resulting deterministic (given the initiator matrix) binary graph can be used to generate a stochastic Kronecker graph, by randomly including edges from the binary graph. Kronecker graphs exhibit a number of realistic graph properties. In particular, their degree distributions follow a multinomial distribution and can exhibit power law behavior depending on the initiator matrix, and for certain parametrizations they yield sparse graphs. However,  they do not generally yield high clustering coefficients \citep{Pinar:Seshadhri:Koldar:2012,Durak:Pinar:Kolda:Seshadhri:2012}, as we see in our experimental analysis (Section~\ref{sec:exp_full}). Further, even in the stochastic setting, Kronecker graphs make very strong self-similarity assumptions about the graph, which may be overly restrictive.

Neural networks have been used to learn generative models for sampling from random graphs \citep{YouYinRenHamLes2018,SimKom2018,LiVinDyPasBat2018}. Typically, these highly parametrized models are learned using multiple graphs.  The sparsity properties of these graphs have not been explored in the literature.

Preferential attachment models \citep{Barabasi:Albert:1999,Albert:Barabasi:2002} are another family of random graphs yield graphs with  power-law degree distribution, where graphs are constructed by adding edges according to a rich-get-richer rule. However, in standard preferential attachment models, the average clustering coefficient tends to zero as the graph grows \citep{Albert:Barabasi:2002,Bollobas:Riordan:2003}. This can be avoided if we have spatial or temporal information about our graph, by incorporating spatial information into the attachment rule \citep{Flaxman:Frieze:Vera:2007,Cooper:Frieze:Pralat:2012,Jacob:Morters:2015}.

\citet{CarFox2017} introduce a family of random graphs or multigraphs parametrized by a generalized gamma process on the space $\mathcal{V}$ of potential vertices, and create either a multigraph by sampling edges as $\mbox{Poisson}(\pi_i\pi_j)$, or a graph by sampling edges as $\mbox{Bernoulli}\left(1-\exp\{-\pi_i\pi_j\}\right)$ (i.e. truncating the multigraph). This can be generalized further by specifying the model in terms of a function $W$ applied to a Poisson process on the space $\mathcal{V}\times\mathbb{R}_+$ \citep{VeiRoy2015,BorChaHol2018}. Under certain choices of $W$ (including those corresponding to the original \citet{CarFox2017} paper), the resulting graphs are sparse and exhibit power law degree distribution. This framework has been extended to incorporate clustering structure \citep{HerSchMor2016,TodMisCar2016}.

The edge-exchangeable graph framework \citep{CaiCamBro2016,CraDem2017} constructs multigraphs that are based on a latent atomic measure $\mathcal{M} = \sum_i\pi_i \delta_{\theta_i}$, where edges can be thought of arriving sequentially according to a distribution parametrized by $\mathcal{M}\times\mathcal{M}$. Under certain conditions on $\mathcal{M}$, the number of vertices grows sub-quadratically with the number of edges as the sequence progresses, and the graph has power law degree distribution. In the specific case where $\mathcal{M}$ is a normalized generalized gamma process, the edge-exchangeable graph is equivalent to the Caron \& Fox multigraph when conditioned on the number of edges. Structure can be incorporated using a mixture of edge-exchangeable graphs \citep{Wil2016}.

Due to their exchangeable construction, neither the sparse exchangeable graphs nor the edge-exchangeable graphs are able to obtain graphs with the sort of clique structure shown here.  This is because the exchangeable constructions do not allow transitivity: conditioned on the latent representations, the presence of edges $(i, j)$ and $(j,k)$ does not increase the probability of the edge $(i, k)$ also being present. In sparse exchangeable setting, while the resulting graphs will contain dense subgraphs, most of the vertices will not belong to such a subgraph \cite{BorChaHol2018,VeiRoy2015}.

One variant of the edge-exchangeable graphs that is particularly relevant to our work is the multiple edges per step edge-exchangeable graph proposed by \cite{CaiCamBro2016}. Here, at each step a random number of edges are added, with edge $(i,j)$ selected with probability $\mu_i\mu_j$ independently of all other edges, leading to multiple edges per step.  If---as is given as an example in \citet{CaiCamBro2016}---the $\mu_i$ are the atom sizes of a stable beta process, then the expected total number of edges and vertices, and the expected number of edges between two vertices conditioned on their associated $\mu_i$, coincide with the model described in this paper, due to linearity of expectation. The difference arises because the multiple edge-exchangeable model generates edges \textit{independently} at each step. In our model, within each clique, the edges are dependent -- giving rise to dense subgraphs. Despite having similar expected sparsity, the multi-step edge-exchangeable model does not have the same local density properties as our model. Further, by representing edge probabilities using a product of beta stable processes, this edge-exchangeable model model lacks the conjugacy that is present when we represent clique inclusion probabilities using a single stable beta process. As a result, posterior inference in the multiple-edge stable beta model poses a significant challenge, and has not yet been addressed in the literature.

\section{EXPERIMENTAL EVALUATION}

We explore properties of our distribution over graphs in both the fully observed and the partially observed setting. We infer the clique matrix underlying a random graph using a reversible jump MCMC algorithm that proposes either splitting or merging cliques. In the partially observed setting, we augment this split/merge procedure with Gibbs steps that help learn the fine-grained structure of the clique cover; we found that adding Gibbs steps had little benefit in the fully observed setting. Full inference details are given in the supplement. 

In our experiments, we optimize the hyperparameters using gradient descent; an alternative approach would be to place a prior on them and infer their posterior distribution using Metropolis-Hastings proposals. 

\subsection{FULLY OBSERVED GRAPH SETTING}\label{sec:exp_full}

We begin by ascertaining how well the fully observed model can capture statistical properties of real-world graphs. We consider the following graph statistics: ratio of triangles to vertices; density $D= 2|E|/|V|(|V|-1)$; average degree; average maximal clique per vertex; average local clustering coefficient. We also look at the degree distribution and the distribution over the maximal clique per vertex. We use the \texttt{NetworkX} python package \citep{HagSwaChu2008} to calculate clustering coefficients and maximal cliques.

\begin{table}[h!]
    \centering
    \footnotesize
    \begin{subtable}{\columnwidth}
        \centering
        \begin{tabular}{c|r|c|c|c}
             &  Truth &  RCC & BNP & Kron  \\  \midrule
            triang./vertex & 1.34 &  $1.99 (.13)$  & $0.01 (.01)$ & $1.21 (.06)$ \\
            density ($\times $1k) & 1.29 &  $1.74 ( .01)$ & $1.28 ( .01)$ & $5.03 (.11) $\\
            av. degree & 4.16 & $4.40 (.14)$ & $3.47(.18)$ & $7.70 (.09)$ \\
            max. clique & 3.50  & $4.20 (.09)$ & $2.03(.01)$ & $2.49 (.03) $\\
            cluster. coeff. & 0.59 & $0.69 ( .02)$  & $0.00(.00)$ & $0.06 (.01) $\\
                    \end{tabular}
        \caption{NeurIPS collaborations.}
    \end{subtable}

    \begin{subtable}{\columnwidth}
        \centering
        \begin{tabular}{c|r|c|c|c}
             & Truth &  RCC & BNP & Kron \\ \midrule
            triang./vertex & 1.91 & $4.14 (.27)$  & $0.01 (.01)$ & $0.01 (.01)$ \\
            density ($\times $1k)  & 2.00  &  $2.70 (.01)$ &  $2.00 (.00)$ & $0.75(.08)$\\
            av. degree &  4.57 & $6.05 (.17)$ & $4.58 (.12)$  & $2.60 (.02)$\\
            max. clique &  4.28 &  $5.52 (.14)$ &  $2.03(.01)$ & $2.03 ( .01)$\\
            cluster. coeff. &  0.81 & $0.78 (.02)$ & $0.00(.00)$ & $0.00(0.01)$\\
        \end{tabular}
        \caption{IMDB comedies.}
    \end{subtable}
    
    \begin{subtable}{\columnwidth}
    \centering
    \begin{tabular}{c|r|c|c|c}
         & Truth &  RCC & BNP & Kron  \\ \midrule
         triang./vertex & 9.21 & $3.25 ( .40)$  &  $0.07 (.01)$ &  $0.01 ( .01)$\\
        density ($\times $1k) & 1.06  & $1.29 ( .01)$ &  $1.07 ( .00)$ & $0.47 ( .01)$\\
        av. degree & 5.53  & $6.53(.47)$ & $5.61(.16)$ & $3.25 ( .01)$\\
        max. clique &  5.13 & $4.64(.15)$ & $2.11(.01)$ & $2.03( .01)$\\
        cluster. coeff. & 0.56 & $0.64 (.02)$ & $0.00(.01)$ & $0.00 (.01)$\\
    \end{tabular}
    \caption{ArXiv GR-QC.}
    \end{subtable}
    
    \begin{subtable}{\columnwidth}
    \centering
    \begin{tabular}{c|r|c|c|c}
         & Truth &  RCC & BNP & Kron  \\ \midrule
        triang./vertex & 2.76 &  $2.26 ( .53)$  & $2.27 (.80)$ & $2.71 (.11)$ \\
        density ($\times $1k)  & 6.30  &  $7.58 (.53)$   & $6.45 (.44)$ & $6.05 (.12)$ \\
        av. degree & 7.33  & $7.58 ( .64)$ & $7.62 ( .68)$ & $9.67 ( .09)$\\
        max. clique &  3.17 & $3.21 (.10)$ & $2.58 (.01)$ & $2.74 (.04)$\\
        cluster. coeff. & 0.22 & $0.25 ( .02)$ & $0.05 (.01)$ & $0.09(.01)$\\
    \end{tabular}
    \caption{ENRON emails.}
    \end{subtable}
    
    \begin{subtable}{\columnwidth}
    \centering
    \begin{tabular}{c|c|c|c|c}
         & Truth &  RCC & BNP & Kron  \\ \midrule
        triang./vertex &  $356$ & $211 (71)$ & $35.1 (7.0)$ & $4.77 ( .50)$ \\
        density ($\times $1k)  & 10.8 &  $12.0 ( .58)$   & $11.1 (.41)$ & \;\,$8.7 (.04)$ \\
        av. degree & 43.7  & $48.0 ( 2.6)$ & $44.9 ( 2.3)$ & $ 35.5 (.01)$\\
        max. clique &  15.9 &  $7.27 (.50)$ &  $3.63 (.09)$  &  $3.11 (.03)$ \\
        cluster. coeff. & 0.60 & $0.42 ( .04)$ & $0.05 (.01)$ & $ 0.02 (.01) $\\
    \end{tabular}
    \caption{Facebook circles.}
    \end{subtable}
    \centering
    
    \caption{Empirical evaluation of statistical properties of random graphs generated from the parameters learned from real-world graphs (standard error in parentheses). The parameters are taken from posterior MCMC draws.}
    \label{tab:experiments}
\end{table}

\begin{figure}[h!]
    \centering
    \begin{subfigure}{\linewidth}
        \centering
        \begin{subfigure}{.45\linewidth}
            \includegraphics[width=\linewidth]{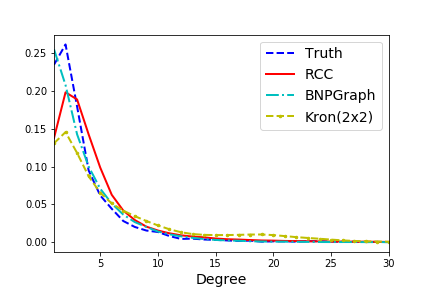}
        \end{subfigure}
        \begin{subfigure}{.45\linewidth}
            \centering
            \includegraphics[width=\linewidth]{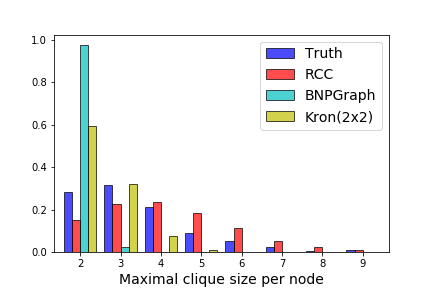}
        \end{subfigure}
        \caption{NeurIPS collaborations (1987-2003).}
    \end{subfigure}
    \begin{subfigure}{\linewidth}
        \centering
        \begin{subfigure}{.45\linewidth}
            \includegraphics[width=\linewidth]{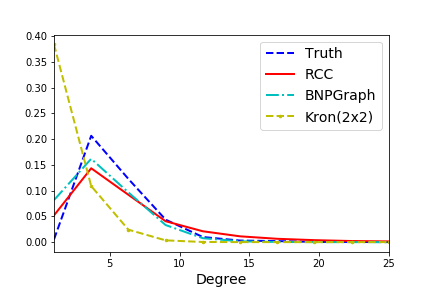}
        \end{subfigure}
        \begin{subfigure}{.45\linewidth}
            \centering
            \includegraphics[width=\linewidth]{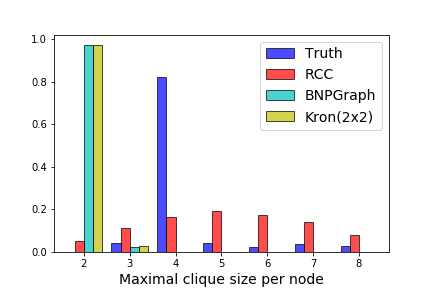}
        \end{subfigure}
        \caption{IMDB comedies cast (2000-2002).}
    \end{subfigure}
    \begin{subfigure}{\linewidth}
        \centering
        \begin{subfigure}{.45\linewidth}
            \includegraphics[width=\linewidth]{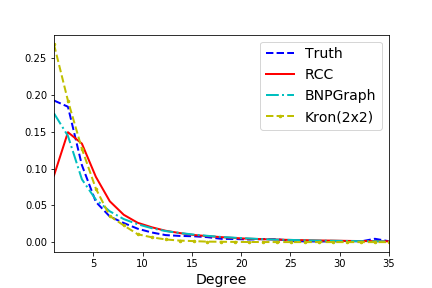}
        \end{subfigure}
        \begin{subfigure}{.45\linewidth}
            \centering
            \includegraphics[width=\linewidth]{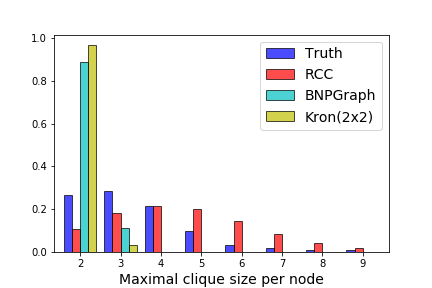}
        \end{subfigure}
        \caption{ArXiv GR-QC (up to 2003).}
    \end{subfigure}
  \begin{subfigure}{\linewidth}
        \centering
        \begin{subfigure}{.45\linewidth}
            \includegraphics[width=\linewidth]{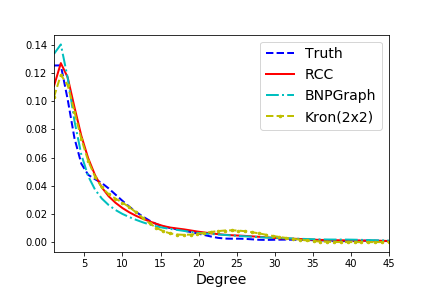}
        \end{subfigure}
        \begin{subfigure}{.45\linewidth}
            \centering
            \includegraphics[width=\linewidth]{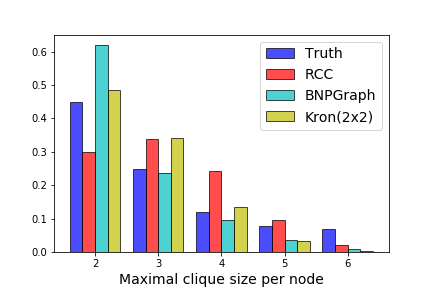}
        \end{subfigure}
        \caption{ENRON (1 month of data)}
    \end{subfigure}
  \begin{subfigure}{\linewidth}
        \centering
        \begin{subfigure}{.45\linewidth}
            \centering
            \includegraphics[width=\linewidth]{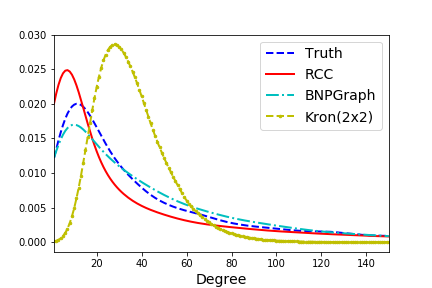}
        \end{subfigure}
           \begin{subfigure}{.45\linewidth}
            \centering
            \includegraphics[width=\linewidth]{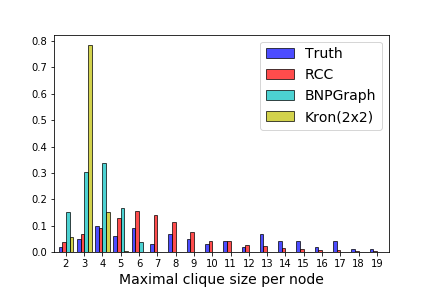}
        \end{subfigure}
        \caption{Facebook circles \citep{facebooknetwork}}
    \end{subfigure}
    \caption{Empirical degree distributions and maximal clique distributions.}\label{fig:degree_and_clique}
\end{figure}

\begin{table*}[ht!]
\centering
\begin{tabular}{ccp{.75\linewidth}}
\toprule
Clique    &  Mean dist. &  Authors                                                       \\ \midrule
     1    &   1.61      & Weare(2), Platt(2), Burges(2), Liblit(1), Aiken(1), Zheng(1), Swenson(2), Crisp(3)\\ 
     2    &   1.86      & Sundararajan(3), Ghaoui(1), Bhattacharyya(1), Lanckriet(1), Keerthi(2), Nilim(2) \\
     3    &   1.91      &  Kim(1), Sastry(1), Xing(1), Weiss(1), Yedidia(2), Ng(1), Russell(1), Freeman(2), Karp(1), Yanover(2) \\
     4    &   2.16     &  Fukumizu(1), Amari(2), Bartlett(1), Wu(3), McAuliffe(1), Nakahara(3), Murata(3), Akaho(2), Bach(1) \\
     5    &   1.82      &  Blei(1), Stromsten(2), DeSilva(2), Kemp(2), Ng(1), Steyvers(2), Griffiths(1), Danks(2), Tenenbaum(1), Sanjana(2) \\
     6    &   2.36      &  Chichilnisky(3), Nguyen(1), Simoncelli(2), Pillow(3), Schwartz(3), Wang(3), Wainwright(1), Todorov(1), Paninski(3) \\ \bottomrule
\end{tabular}
\caption{Meta-analysis of the six cliques of Michael Jordan. Shortest-path distance to Michael Jordan is shown in parentheses; connections not present in the original graph are shown in red. Mean dist. is the mean path length (along the original graph) between clique members.}
\label{tab:partial_cliques}
\end{table*}

To determine how well-suited our model is for modeling real graphs, we consider four real-world graphs,
\begin{itemize}
\item The co-authorship graph for NeurIPS publications from 1987-2003 (2715 vertices, 4733 edges).\footnote{\url{https://www.kaggle.com/benhamner/nips-papers}}
\item The largest connected component of the co-appearance graph generated from IMDB cast lists of comedies from 2000-2002 (2288 vertices, 5232 vertices).\footnote{\url{www.imdb.com}}
\item The co-authorship graph for arXiv publications in the category GR-QC (General Relativity and Quantum Cosmology) from 2003 (5241 vertices, 14484 edges) \citep{LesKleFal2007}.
\item The full interaction graph from a single month (January 2000) of the ENRON dataset (1172 vertices, 4293 edges) \citep{KliYa2004}.
\item Anonymous data collected from Facebook social circles,\footnote{The computation of the maximal clique per vertex for the Facebook social circles was omitted due to the size of this graph.} comprising 10 interconnected friend lists of people obtained using an external app (4039 vertices, 88234 edges) \citep{facebooknetwork}.
\end{itemize}

We model each graph using our proposed (fully observed) random clique cover model---labeled RCC in tables and figures---and learn the hyperparameters using maximum likelihood. We then sample 25 graphs using these hyperparameters, and compare the average statistics of the sampled graph with those of the original graph in Table~\ref{tab:experiments} and Figure~\ref{fig:degree_and_clique}. We consider two comparison models:
\begin{itemize}
    \item \textbf{BNPgraph}: The sparse exchangeable model of \citet{CarFox2017} (implemented using code published by the authors), which can capture graph sparsity,   power-law degree distribution, and an unbounded number of vertices. 
    \item \textbf{Kron}: A stochastic Kronecker graph, with 2$\times$2 initiator matrix inferred and sampled using the \texttt{kronfit} and \texttt{krongen} software packages introduced by \citet{Leskovec:Chakrabarti:Kleinberg:Faloutsos:Ghahramani:2010}. 2$\times$2 initiator matrices have been shown to be a good fit for real-world graph models \citep{Leskovec:Chakrabarti:Kleinberg:Faloutsos:Ghahramani:2010}.
\end{itemize}

We see that both the Caron \& Fox model and our model models do a good job at capturing sparsity and degree distribution. However, the Caron \& Fox model does not capture the locally dense structure. While simulations from our model have similar clustering coefficients to the real graphs, the Caron \& Fox model has near-zero clustering coefficients. Our simulations also exhibit larger maximal cliques. While the maximal clique distribution for the IMDB dataset seems to overestimate the number of cliques, we believe this is an artifact of the data. The dataset only includes the four highest billed actors for each movie, artificially deflating clique sizes. The Kronecker graph generally does a worse job at capturing the graph sparsity and degree, which we hypothesize is due to the restrictive self-similarity of the underlying graph.

We further inspected the differences between models by looking at the adjacency matrices and visual representations of the random graphs generated from the best-fitted parameters. Figure \ref{fig:qualitcomparison} shows the results. For ease of visualization, we only include the top-left $100\times 100$ block of the adjacency matrices; the full matrices are included in the supplement. The sparse exchangeable graph is sparse everywhere with no community structure, whereas the Kronecker graph shows a structure unrelated to the original graph. In contrast, our model produces a graph that looks more similar to the original graph, and with a similar community structure shown by the adjacency matrix. We used the Python package \texttt{NetworkX} \citep{HagSwaChu2008} to draw the graphs with a spring layout; we also used its implementation of the Girvan-Newman algorithm \citep{Girvan:Newman:2002} to arrange the adjacency matrices by communities.

\begin{figure}[h!]
    \centering
    \begin{subfigure}{\linewidth}
        \centering
        \begin{subfigure}{.59\linewidth}
            \includegraphics[width=\linewidth]{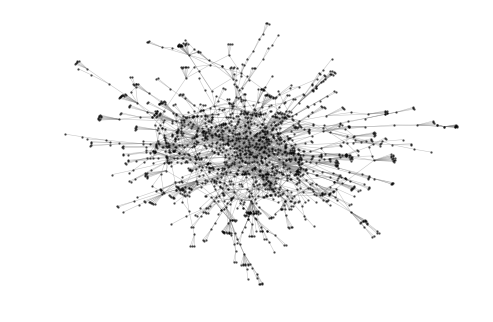}
        \end{subfigure}
        \begin{subfigure}{.39\linewidth}
            \centering
            \includegraphics[width=\linewidth]{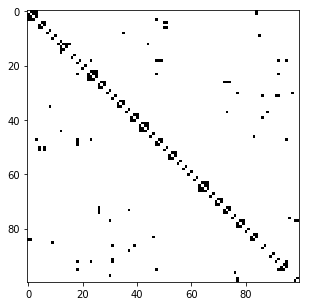}
        \end{subfigure}
        \caption{Original NeurIPS graph.}
    \end{subfigure}
    \begin{subfigure}{\linewidth}
        \centering
        \begin{subfigure}{.59\linewidth}
            \includegraphics[width=\linewidth]{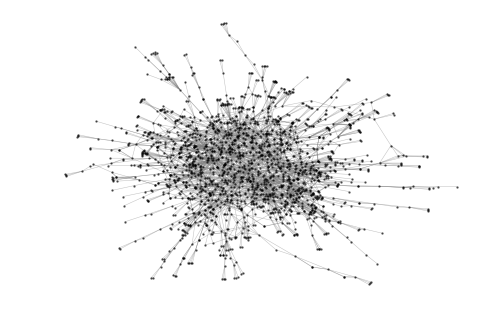}
        \end{subfigure}
        \begin{subfigure}{.39\linewidth}
            \centering
            \includegraphics[width=\linewidth]{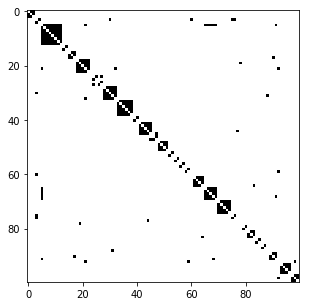}
        \end{subfigure}
        \caption{Random Clique Cover (RCC).}
    \end{subfigure}
    \begin{subfigure}{\linewidth}
        \centering
        \begin{subfigure}{.59\linewidth}
            \includegraphics[width=\linewidth]{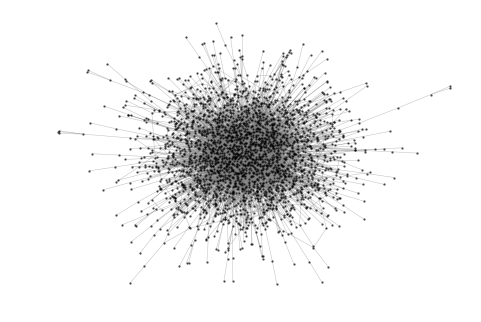}
        \end{subfigure}
        \begin{subfigure}{.39\linewidth}
            \centering
            \includegraphics[width=\linewidth]{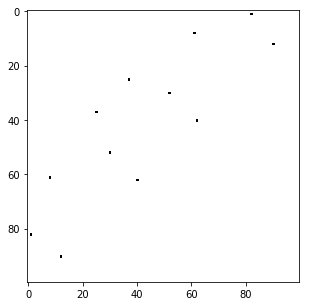}
        \end{subfigure}
        \caption{Sparse exchangeable graph (BNPgraph).}
    \end{subfigure}
    \begin{subfigure}{\linewidth}
        \centering
        \begin{subfigure}{.59\linewidth}
            \includegraphics[width=\linewidth]{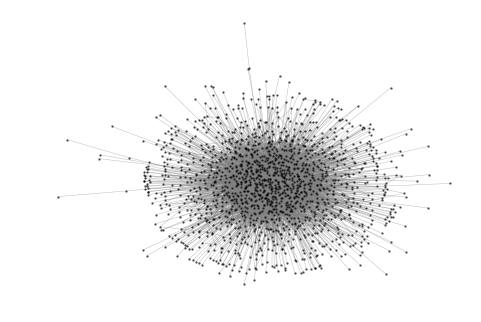}
        \end{subfigure}
        \begin{subfigure}{.39\linewidth}
            \centering
            \includegraphics[width=\linewidth]{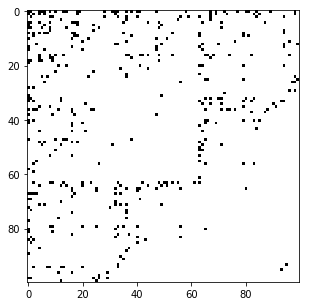}
        \end{subfigure}
        \caption{Kronecker graph with 2$\times$2 initiator matrix (Kron).}
    \end{subfigure}
    \caption{Co-authorship  graph  for NeurIPS  publications (1987-2003) and graphs sampled using (b) our model, and (c, d) comparison methods. Left: graph visualization; right: top-left 100x100 block of the adjacency matrix arranged by communities. }\label{fig:qualitcomparison}
\end{figure}

\subsection{PARTIALLY OBSERVED GRAPH SETTING}\label{sec:exp_partial}
\begin{figure}[h!]
    \centering
    \begin{subfigure}{.7\linewidth}
        \includegraphics[width=\linewidth]{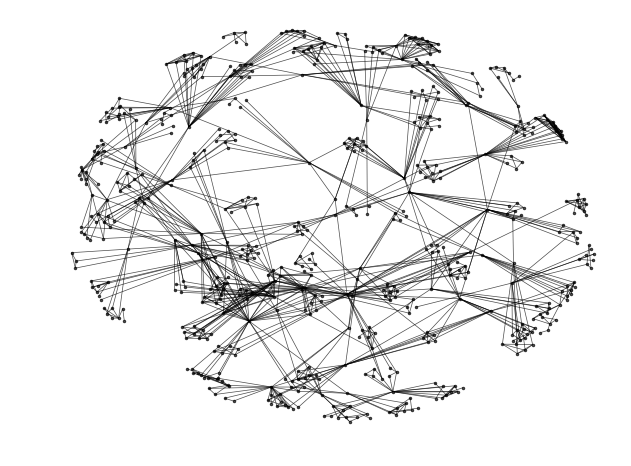}
        \caption{Original edges.}
    \end{subfigure}\hfill
    \begin{subfigure}{.7\linewidth}
        \includegraphics[width=\linewidth]{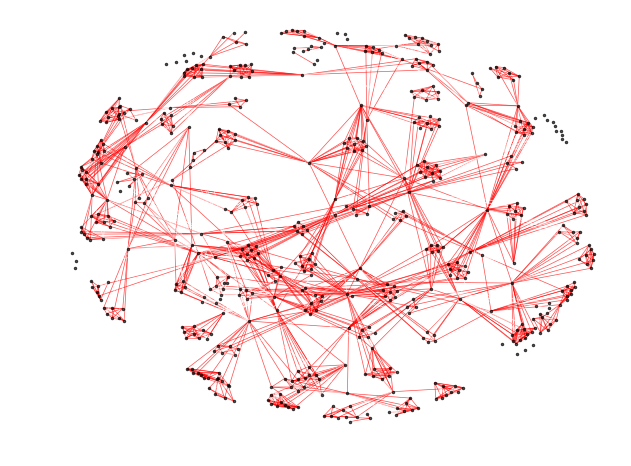}
        \caption{Inferred latent edges.}
    \end{subfigure}\hfill
   \begin{subfigure}{.7\linewidth}
    \includegraphics[width=\linewidth]{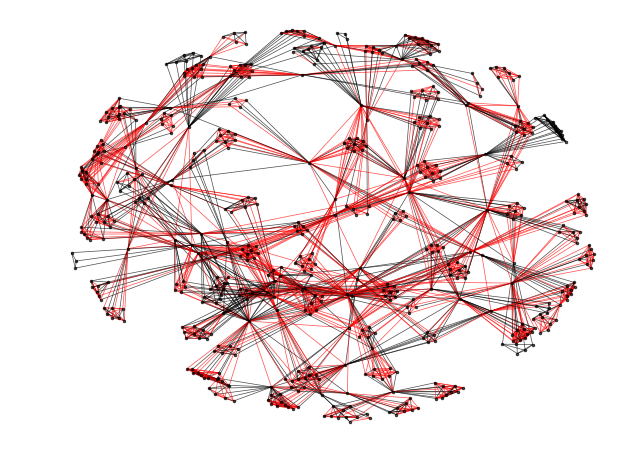}
    \caption{Full inferred latent graph.}
    \end{subfigure}
    \caption{Partially observed graph inference for NeurIPS co-authorships (1999-2003), edges appearing in the latent cliques but not the original graph are shown in red}
    \label{fig:nips-orig-partial}
\end{figure}

The partially observed setting offers the ability to infer a latent graph---and corresponding set of latent cliques---underlying an observed graph. Intuitively, this latent graph is likely to capture latent community-type structure. To evaluate this empirically, we modeled the NeurIPS co-authorship dataset described above using a partially observed model. We consider publications within the period 1999-2003, assume a shared clique parameter $\pi$, and infer hyperparameters using Metropolis-Hastings.  Figure \ref{fig:nips-orig-partial} shows the latent structure found (based on a single sample from the posterior). 

Posterior draws for $\pi$ concentrated around 0.4, robustly with respect to hyperparameter specification. Thus, roughly 60\% of the edges in the inferred graph are not in the original graph. As expected given the short period of analysis, 491 out of the 553 authors in consideration belonged to only one clique, and only eleven of them to more than three. The average clique size was 8.93. The author belonging to the most cliques was (perhaps unsurprisingly) Michael Jordan. The members of the six cliques he belonged to are shown in Table~\ref{tab:partial_cliques}. It is worth mentioning that case-by-case inspection show that in some cases, inferred edges were in fact collaborations outside the period of observation. 
\section{DISCUSSION}

We have presented a new class of Bayesian nonparametric prior for graphs, based on random clique selection mechanisms, that are appropriate for many real-world graphs. We have presented some preliminary work in a modeling context; we hope to see this form of random graph-based model explored further in future.

While we base our model on the stable beta process, alternative subset selection mechanisms could be proposed. For example, a restricted stable beta IBP \citep{WilMacXin2013,UtkPraStoPerKoc2018} could give more explicit control over the number of cliques. We leave exploration of alternative mechanisms as future work.

\bibliographystyle{plainnat}
\bibliography{sparsedense}

\begin{appendices}
\appendix

\section{RANDOM MULTIGRAPHS}\label{app:multigraph}
Our random clique construction can also be used to define a multigraph, where the number of edges between vertices $i$ and $j$ is given by the number of generating cliques containing both vertices. The adjacency matrix of this multigraph is given by

$$Z^TZ - \mbox{diag}(X^TZ)$$

where $Z$ is the binary matrix where $Z_{ni}=1$ iff clique $n$ contains vertex $i$.

This multigraph exhibits similar properties to the graph. We can define multigraph sparsity analogously to graph sparsity, saying a random multigraph is sparse if the number of edges increases sub-quadratically with the number of vertices.  Proposition~\ref{prop:vertices} holds in the multigraph setting, indicating that the number of vertices  $|V|$ grows with the number of cliques $\nc$ as $\nc^\sigma$. Since the number of vertices in each clique is marginally distributed as $\mbox{Poisson}(\alpha)$, the expected number of edges in $\nc$ cliques is $\nc \alpha^2/2$ -- i.e.\ $|E|$ grows linearly in expectation with $\nc$. This implies that the number of edges grows with the number of vertices as  $O(|V|^{1/\sigma})$, indicating the multigraph will be sparse when $\sigma>0.5$. We demonstrate this empirically in Appendix~\ref{app:sparsity_sims}.

\section{INFERENCE}

\paragraph{Initialization}
In the fully observed graph setting, the likelihood $P(G|Z)$ is a step function that is one iff $G = \min(Z^TZ - \mbox{diag}(Z^TZ),1)$, so a sampler can only make moves that do not change the sparsity pattern of $Z^TZ$. Because of this, the clique must be initialized to values compatible with $G$. One way of doing this is to initialize to the set of 2-cliques corresponding to the edges of the graph. An alternative approach is to use an edge clique cover algorithm; however this may be infeasible for large graphs since finding a minimal cover is NP hard. For the partially observed graph, it suffices to choose an initialization so that $Z^TZ$ covers all edges in the graph.

\paragraph{Split/merge sampling for $Z$} To explore the space of clique covers, we propose split/merge moves that can either split a large clique into two potentially overlapping cliques, or merge two cliques into a single clique. We select an edge $(i, j)$ uniformly at random from $E$, and then for each vertex $i,j$ associated with that edge, we select a clique index $\cc_i, \cc_j$ uniformly from the set of cliques it belongs to. 
If $\cc_i = \cc_j$, we propose a new matrix $Z^*$, where the row $Z_{\cc_i}$ has been replaced with two rows $Z^*_{\cc_i}$ and $Z^*_{\cc_j}$. We set $Z^*_{\cc_{i}i}=1$ and $Z^*_{\cc_{j}j}=1$. Then, for each non-zero entry $k$ in $Z_{\cc_i}$, we consider possible settings of $(Z^*_{\cc_{i}k}, Z^*_{\cc_{j}k}) \in \{(0,1),(1,0),(1,1)\}$. We exclude any settings that are incompatible with the network structure, and select uniformly between the other options.

If $\cc_i\neq \cc_j$, we propose replacing the cliques $Z_{\cc_i}$ and $Z_{\cc_j}$ with their union $Z^*_{\cc_i} = Z_{\cc_i}\cup Z_{\cc_j}$. We note that the resulting proposal could have a likelihood of zero, due to introducing edges that do not appear in the network.

We accept or reject the split or merge by calculating the appropriate reversible jump MCMC acceptance probability.

\paragraph{Gibbs sampling step for $Z_{ij}$} 
For $Z_{ij}$ such that $m_{k}^{\neg i} = \sum_{i^\prime \neq i} Z_{i^\prime j} <0$, we can augment the split/merge proposals with Gibbs steps for that sample from the conditional distribution
$$\begin{aligned}
p(Z_{ij} = 1|Z_{\neg ij}, G ) \propto& \frac{m_j^{\neg i} - \sigma}{\nc +c -1} \mathcal{L}(Z_{ij} = 1)\\
p(Z_{ij} = 0|Z_{\neg ij}, G ) \propto& \frac{\nc+c-1-m_j^{\neg i} + \sigma}{\nc+c -1}\mathcal{L}(Z_{ij} = 1) ,\end{aligned}$$
where $\mathcal{L}(Z_{ij} = z) = P(G|Z_{ij} = z, Z_{\neg ij}, \{\pi_cc\})$ is the likelihood of the graph given the proposed clique matrix.

We can then sample the number of all-zero cliques using a Metropolis-Hastings proposal, proposing a new number of all-zero cliques from an appropriate distribution and calculating the corresponding acceptance probability. We found adding Gibbs sampling steps did not improve mixing in the fully observed setting (where the set of compatible clique matrices is more constrained), but was important in the partially observed setting.

\section{Simulation results}\label{app:sparsity_sims}
In this section, we provide simulations that support the theoretical properties explored in the main paper.

In Section~\ref{sec:sparsity} we showed that the number of edges grows sub-quadratically with the number of vertices, as $0\left(\min(|V|^{1+\sigma}{2\sigma}, |V|^\frac{3}{2})\right)$. In Appendix~\ref{app:multigraph}, we showed that in the related multigraph, the number of edges grows as $O(|V|^{1/\sigma})$, which is sub-quadratic when $\sigma>0.5$.

We can validate this limiting behavior using simulations from the graph and multigraph. Figures~\ref{fig:g_sparsity} and \ref{fig:mg_sparsity}  show how $|E|$ and $|V|$ co-vary in random graphs and random multigraphs respectively, for different values of $c$ and $\sigma$.  The scatter plots show $(|V|,|E|)$ pairs for ten simulations, evaluated after each of 100 cliques are added. In the random graphs (Figure~\ref{fig:g_sparsity}), for high values of $\sigma$ we have a near-linear relationship between $|V|$ and $|E|$, indicating an extremely sparse graph. As $\sigma$ decreases the exponent increases and the level of sparsity decreases (although the graph never becomes dense). For the multigraphs (Figure~\ref{fig:mg_sparsity}) we see a sub-quadratic relationship between $|V|$ and $|E|$ for $\sigma>0.5$, but not for $\sigma \leq 0.5$. In each case, the concentration parameter $c$ controls the variance.

\begin{figure*}
  \centering
  \begin{subfigure}{0.3\textwidth}
    \includegraphics[width=\textwidth]{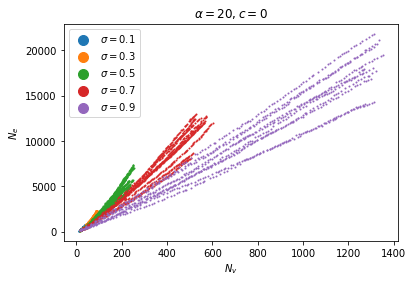}
  \end{subfigure}
  \begin{subfigure}{0.3\textwidth}
    \includegraphics[width=\textwidth]{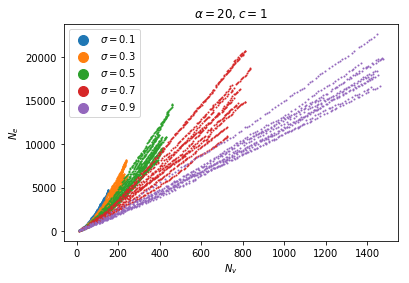}
  \end{subfigure}
  \begin{subfigure}{0.3\textwidth}
    \includegraphics[width=\textwidth]{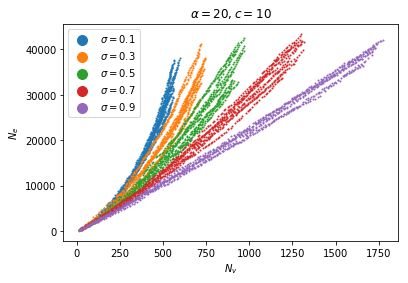}
  \end{subfigure}
  \\
  \begin{subfigure}{0.3\textwidth}
    \includegraphics[width=\textwidth]{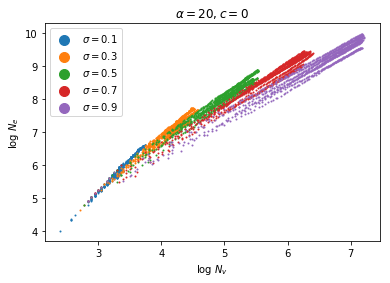}
  \end{subfigure}
  \begin{subfigure}{0.3\textwidth}
    \includegraphics[width=\textwidth]{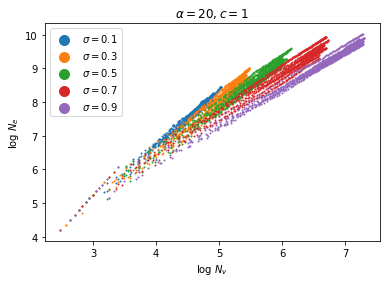}
  \end{subfigure}
  \begin{subfigure}{0.3\textwidth}
    \includegraphics[width=\textwidth]{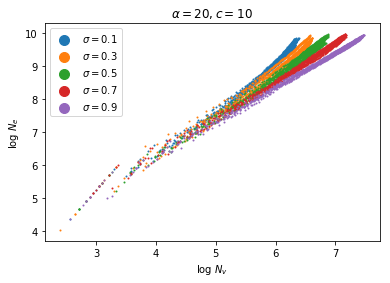}
  \end{subfigure}
  \caption{Number of edges vs number of vertices for a  graph simulated with $\alpha=20$ and varying values for $c$ and $\sigma$. Top: linear scale; bottom: log scale.}
  \label{fig:g_sparsity}
\end{figure*}

\begin{figure*}
  \centering
  \begin{subfigure}{0.3\textwidth}
    \includegraphics[width=\textwidth]{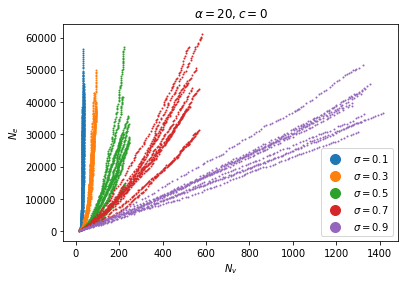}
  \end{subfigure}
  \begin{subfigure}{0.3\textwidth}
    \includegraphics[width=\textwidth]{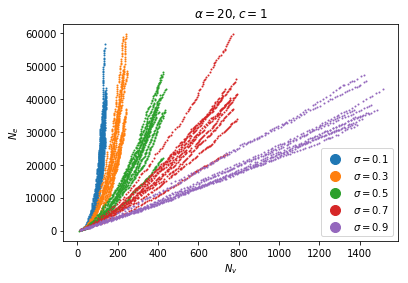}
  \end{subfigure}
  \begin{subfigure}{0.3\textwidth}
    \includegraphics[width=\textwidth]{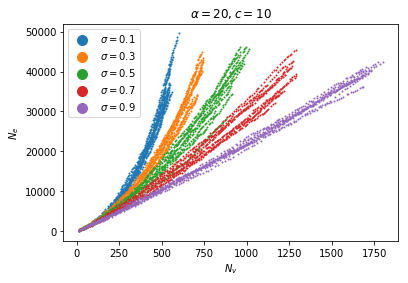}
  \end{subfigure}
  \\
  \begin{subfigure}{0.3\textwidth}
    \includegraphics[width=\textwidth]{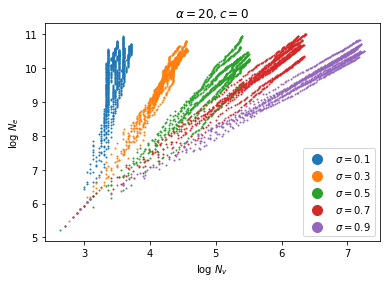}
  \end{subfigure}
  \begin{subfigure}{0.3\textwidth}
    \includegraphics[width=\textwidth]{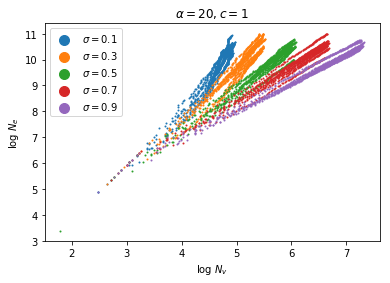}
  \end{subfigure}
  \begin{subfigure}{0.3\textwidth}
    \includegraphics[width=\textwidth]{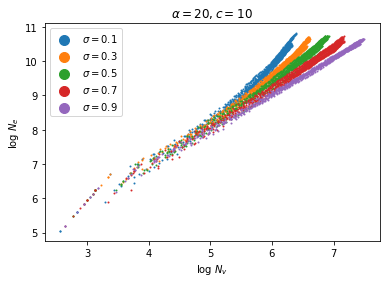}
  \end{subfigure}
  \caption{Number of edges vs number of vertices for a  multigraph simulated with $\alpha=20$ and varying values for $c$ and $\sigma$. Top: linear scale; bottom: log scale.}
  \label{fig:mg_sparsity}
\end{figure*}

We also see this behavior if we look at how the average density $D=2|E|/|V|(|V|-1)$ varies with the number of cliques and the average clique size. Figure~\ref{fig:g_density} shows that in the graph setting, the density decreases as we add more cliques, but converges to a level that depends on the parameter $\sigma$. Figure~\ref{fig:dens_vs_sigma} shows that this value is largely invariant to the clique size. In the multigraph setting, Figure~\ref{fig:mg_density} shows that the density decreases as we add more cliques if $\sigma>0.5$, but increases if $\sigma<0.5$.

\begin{figure*}
  \centering
  \begin{subfigure}{0.46\textwidth}
    \includegraphics[width=\textwidth]{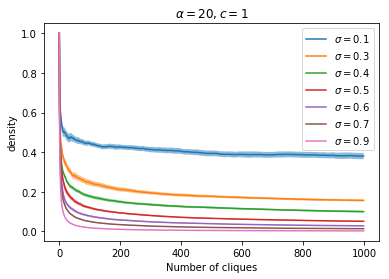}
  \end{subfigure}
  \begin{subfigure}{0.46\textwidth}
    \includegraphics[width=\textwidth]{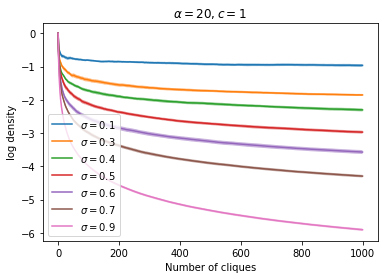}
  \end{subfigure}
  \caption{Network densities for graphs with varying values of $\sigma$, as the number of cliques increases.}
  \label{fig:g_density}
\end{figure*}

\begin{figure*}[t]
    \centering
    \includegraphics[width=.46\textwidth]{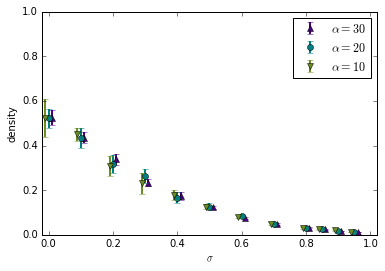}
    \caption{Graph density for increasing values of $\sigma$, for $c=1$, $\nc=100$, and varying values of $\alpha$}
    \label{fig:dens_vs_sigma}
\end{figure*}

\begin{figure*}
  \centering
  \begin{subfigure}{0.46\textwidth}
    \includegraphics[width=\textwidth]{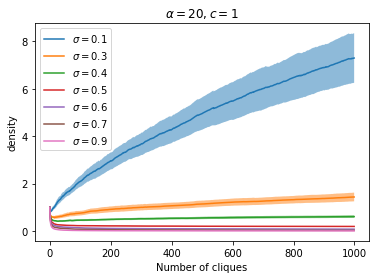}
  \end{subfigure}
  \begin{subfigure}{0.46\textwidth}
    \includegraphics[width=\textwidth]{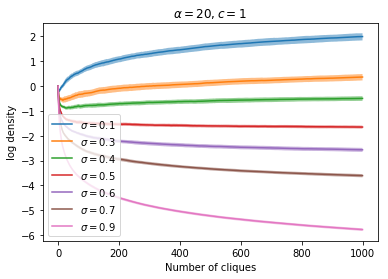}
  \end{subfigure}
  \caption{Network densities for multigraphs with varying values of $\sigma$, as the number of cliques increases.}
  \label{fig:mg_density}
\end{figure*}

In Section~\ref{sec:density}, we argued that the expected average maximal clique is lower-bounded by $\alpha$, and that it will increase as we see greater overlap between cliques. We see this empirically in Figure~\ref{fig:max_clique_vs_sigma}, where we simulated graphs with varying values of $\alpha$ and $\sigma$ and calculated the average maximal clique per vertex using the \texttt{NetworkX} Python package \citep{HagSwaChu2008}.

\begin{figure*}[t]
    \centering
    \includegraphics[width=.5\textwidth]{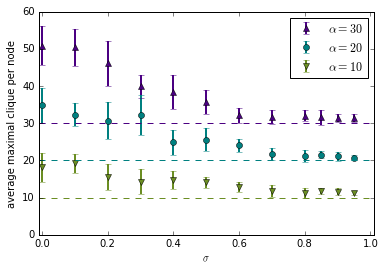}
    \caption{Average largest maximal clique per vertex, for a binary graph with $c=1$, $\nc=100$, and varying values of $\alpha$.}
    \label{fig:max_clique_vs_sigma}
\end{figure*}

\section{Comparison of random adjacency matrices}

In Section~\ref{sec:exp_full}, we showed part of the adjacency matrix for the NeurIPS interaction dataset, and the corresponding adjacency matrices sampled using our model (RCC), the Caron \& Fox sparse exchangeable model (BNPgraph), and a Kronecker graph with 2$\times$2 initiator matrix (Kron). The full adjacency matrices are shown in Figures~\ref{fig:adjmatfull:originalnips}, \ref{fig:adjmatfull:rcc}, \ref{fig:adjmatfull:bnp} and \ref{fig:adjmatfull:kron}. To make the differences clearer, we used the Girvan-Newman algorithm \citep{Girvan:Newman:2002}, as implemented in \texttt{NetworkX} \citep{HagSwaChu2008}, to find communities in the generated graphs. We have added blue boxes to the figures to represent the communities found.

\begin{figure*}
    \centering
    \includegraphics[width=\textwidth]{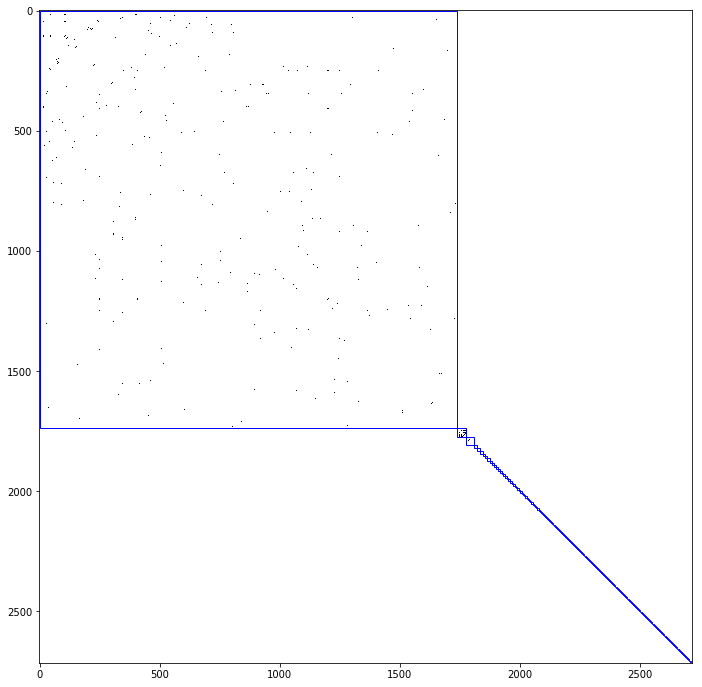}
    \caption{Original co-authorship  adjacency matrix  for  NeurIPS  publications (1987-2003).}
    \label{fig:adjmatfull:originalnips}
\end{figure*}

\begin{figure*}
    \centering
    \includegraphics[width=\textwidth]{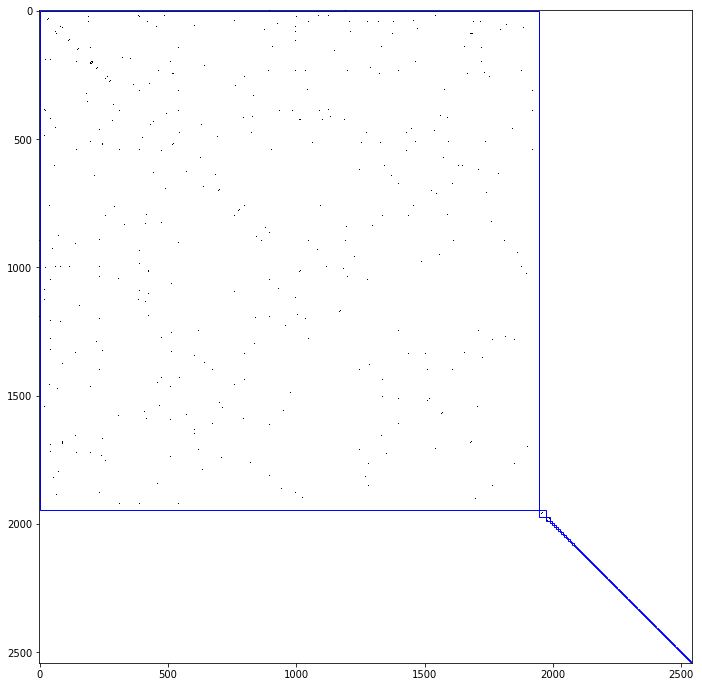}
    \caption{Adjacency matrix of a sample from a random clique cover (RCC) graph, with parameters learned on the NeurIPS dataset.}
    \label{fig:adjmatfull:rcc}
\end{figure*}

\begin{figure*}
    \centering
    \includegraphics[width=\textwidth]{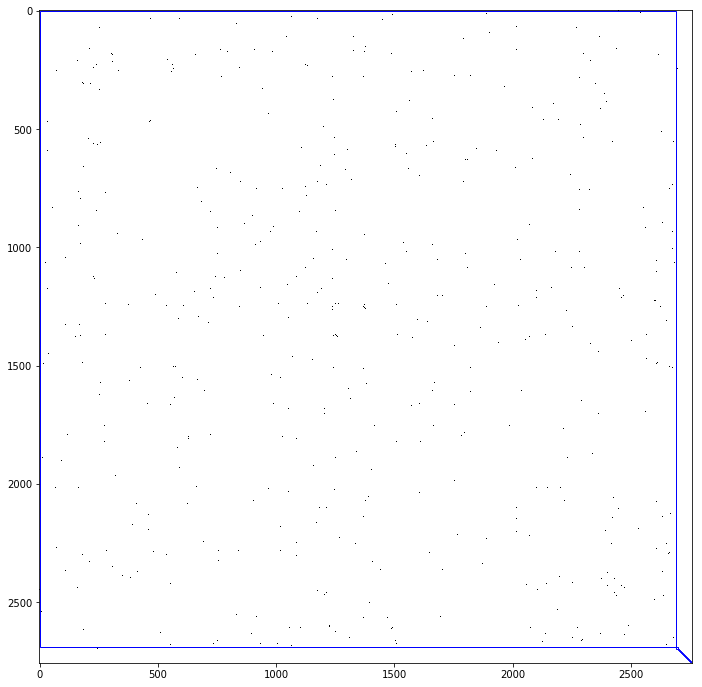}
    \caption{Adjacency matrix of a sample from a sparse exchangeable graph (BNPgraph), with parameters learned on the NeurIPS dataset.}
    \label{fig:adjmatfull:bnp}
\end{figure*}

\begin{figure*}
    \centering
    \includegraphics[width=\textwidth]{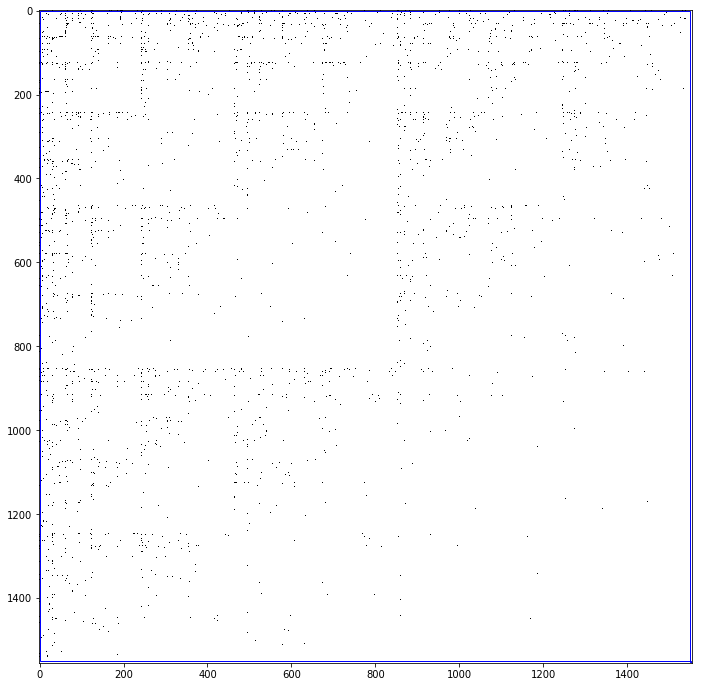}
    \caption{Adjacency matrix of a sample from a Kronecker graph with 2$\times$2 initiator matrix (Kron), with parameters learned on the NeurIPS dataset.}
    \label{fig:adjmatfull:kron}
\end{figure*}

\end{appendices}

\end{document}